\documentclass[11pt]{article}

\usepackage[T1]{fontenc}
\usepackage[margin=1in]{geometry}
\usepackage{amsmath,amssymb,amsthm,mathtools}
\usepackage{algorithm}
\usepackage{algpseudocode}
\usepackage{array}
\usepackage{booktabs}
\usepackage{microtype}
\usepackage{tikz}
\usetikzlibrary{arrows.meta,calc,fit,positioning}
\usepackage{hyperref}
\hypersetup{
  hidelinks,
  pdftitle={Moore's Greedy Algorithm for Minimizing  the Number of Late Jobs:
 Structure and Implementation},
  pdfauthor={Dean Matthew Menezes and C. Gregory Plaxton}
}

\allowdisplaybreaks[2]
\numberwithin{equation}{section}

\newtheorem{lemma}{Lemma}[section]
\newtheorem{theorem}{Theorem}

\newtheorem{corollary}[theorem]{Corollary} 

\DeclareMathOperator{\GreedyOp}{greedy}
\DeclareMathOperator{\EdfOp}{edf}
\DeclareMathOperator{\SjfOp}{sjf}
\DeclareMathOperator{\BigOp}{big}
\DeclareMathOperator{\ShortOp}{short}
\DeclareMathOperator{\FeasibleOp}{feasible}
\DeclareMathOperator{\LexOp}{lex}
\DeclareMathOperator{\ProcOp}{work}
\DeclareMathOperator{\GapOp}{gap}
\DeclareMathOperator{\SlackOp}{slack}
\DeclareMathOperator{\LevelOp}{level}
\DeclareMathOperator{\RankOp}{rank}
\DeclareMathOperator{\IdsOp}{ids}
\DeclareMathOperator{\AddOp}{add}
\DeclareMathOperator{\AdmitsOp}{admits}
\DeclareMathOperator{\KeysOp}{keys}

\DeclareMathOperator*{\ArgMax}{arg\,max}
\DeclareMathOperator*{\ArgMin}{arg\,min}

\newcommand{\Prob}{1\mid\mid\sum U_j}

\newcommand{\PosReals}{\mathbb{R}_+}
\newcommand{\Seq}{\sigma}
\newcommand{\LessThan}[2]{#1^{<#2}}
\newcommand{\EqualTo}[2]{#1^{=#2}}
\newcommand{\AtMost}[2]{#1^{\leq #2}}
\newcommand{\Edf}[1]{\EdfOp(#1)}
\newcommand{\Sjf}[1]{\SjfOp(#1)}
\newcommand{\MyBig}[1]{\BigOp(#1)}
\newcommand{\Short}[1]{\ShortOp(#1)}
\newcommand{\ShortTwo}[2]{\ShortOp(#1,#2)}
\newcommand{\Feasible}[1]{\FeasibleOp(#1)}
\newcommand{\Lex}[1]{\LexOp(#1)}
\newcommand{\Proc}[1]{\ProcOp(#1)}
\newcommand{\ProcTwo}[2]{\ProcOp(#1,#2)}
\newcommand{\Gap}[2]{\GapOp(#1,#2)}
\newcommand{\Slack}[1]{\SlackOp(#1)}
\newcommand{\SlackTwo}[2]{\SlackOp(#1,#2)}
\newcommand{\Level}[2]{\LevelOp(#1,#2)}
\newcommand{\Matroid}{M}
\newcommand{\Greedy}[1]{\GreedyOp(#1)}
\newcommand{\GreedyTwo}[2]{\GreedyOp(#1,#2)}
\newcommand{\SameDist}{\simeq}
\newcommand{\Rank}[1]{\RankOp(#1)}
\newcommand{\Ids}[1]{\IdsOp(#1)}
\newcommand{\TreeAdd}[3]{#1.\AddOp(#2,#3)}
\newcommand{\TreeAdmits}[3]{#1.\AdmitsOp(#2,#3)}
\newcommand{\TreeSlack}[1]{#1.\SlackOp()}

\newcommand{\True}{\mathbf{true}}
\newcommand{\False}{\mathbf{false}}

\newcommand{\Nil}{\mathbf{nil}}
\newcommand{\Keys}[1]{\KeysOp(#1)}

\newcommand{\Ind}{\mathcal I}
\newcommand{\Bases}{\mathcal B}

\title{Moore's Greedy Algorithm for Minimizing the \\Number of Late Jobs:
 Structure and Implementation}
\author{
  Dean Matthew Menezes\\
  University of Texas at Austin\\
  \texttt{dean.menezes@utexas.edu}
  \and
  C. Gregory Plaxton\\
  University of Texas at Austin\\
  \texttt{plaxton@cs.utexas.edu}
}
\date{}

\begin{document}
\maketitle

\begin{abstract}
The problem of scheduling $n$ jobs, each with a deadline and a
processing time, on a single resource to minimize the number
of late jobs is one of the earliest problems considered in the
scheduling literature. It is equivalent to finding a
maximum-cardinality feasible subset of the jobs, where a subset is
feasible if there is a schedule under which all of the associated jobs
meet their deadlines.  Moore's 1968 paper presents two algorithms for
the problem, including an elegant non-greedy $O(n\log n)$-time algorithm
attributed to Hodgson. Moore also presents a
natural greedy algorithm that considers the jobs in nondecreasing
order of processing time, adding each successive job to the current
feasible set whenever the resulting set remains feasible.  Zhao and
Yuan recently obtained the first subquadratic implementation of
Moore's greedy algorithm: They provided a data structure that performs
each feasibility query in amortized $O(\log n)$ time, thereby matching
the overall $O(n\log n)$ time bound of Moore-Hodgson. Our first main
contribution is to provide a simpler augmented BST data structure that
performs each feasibility query in worst-case $O(\log n)$ time.

Moore's greedy algorithm is known to produce a maximum-cardinality
feasible subset with minimum total processing time. The set system
associated with the set of all feasible subsets is not a matroid, but
if it were, Moore's greedy algorithm would be a special case of the
matroid greedy algorithm, yielding the latter property as an immediate
consequence.  In our second main contribution, we explain why the
matroid greedy algorithm works correctly in this non-matroid setting:
We show that if all of the jobs have equal processing time, then the
problem corresponds to a nested matroid, and that when there are
distinct ``tiers'' of processing times, the matroid greedy algorithm
solves a sequence of nested matroid problems, one for each tier.

In our third main contribution, we present an explicit linear-size
flow network that defines a polymatroid rank. A set is
feasible exactly when its rank equals its total processing time, and
Moore accepts a job exactly when it contributes its full
processing-time marginal. Contracting the shorter accepted jobs and
measuring residual rank in units of the current processing time
recovers the complete rank function of the corresponding tier
matroid. As such, the flow rank provides another structural
explanation for the correctness of Moore's greedy algorithm.
\end{abstract}

\section{Introduction}
\label{sec:intro}

This paper is concerned with the following classic scheduling problem,
often denoted $\Prob$ in the scheduling literature: Given a set of jobs
to be scheduled on a single resource beginning at time zero, where each
job is characterized by a positive real deadline and a positive real
processing time, determine an order in which to execute the jobs that
minimizes the number of late jobs.  We say that a set of jobs is
feasible if there is an order of execution for the jobs such that every
job meets its deadline.  Given a set of jobs $S$, the problem of finding
a schedule for $S$ that minimizes the number of late jobs is equivalent
to the problem of finding a maximum-cardinality feasible subset of $S$.
Jackson's rule implies that a set of jobs is feasible if and only if
every job meets its deadline under any earliest-deadline-first (EDF)
schedule; see, for example, Pinedo~\cite{Pinedo22}.

Moore~\cite{Moore68} presents two algorithms for computing a
maximum-cardinality feasible subset $X$ of a given set of jobs $S$.
The first uses a simple greedy approach to grow $X$ from the empty set
by considering each job $\alpha$ in $S$ in nondecreasing order of its
processing time, breaking ties arbitrarily, and adding $\alpha$ to $X$
whenever the resulting set remains feasible.\footnote{Moore's
description of the algorithm is somewhat more involved.  At a high
level, when considering job $\alpha$, Moore first checks a
computationally easy, constant-time condition that is sufficient but
not necessary for $X+\alpha$ to be feasible.  If this condition is not
met, then a more expensive, linear-time computation is performed to
determine whether $X+\alpha$ is feasible.  The overall effect is that
$\alpha$ is added to $X$ if and only if $X+\alpha$ is feasible.  The
simplified description of Moore's greedy algorithm used in our paper
is explicitly stated by Lin and Wang~\cite{LinWang07}.}

Moore's paper includes a brief supplement describing a second
algorithm attributed to T.~J.~Hodgson; in the present paper, we refer
to the latter algorithm as the Moore--Hodgson algorithm.  We maintain
an EDF schedule of an initially empty set $X$ of tentatively accepted
jobs and process all jobs in deadline order, breaking ties
arbitrarily.  To process a job $\alpha$, we append it to the EDF
schedule for $X$.  If the resulting set is infeasible, we delete a job
of maximum processing time from $X+\alpha$.  The Moore-Hodgson
algorithm admits a simple $O(n\log n)$-time implementation, a
significant improvement over the $O(n^2)$-time bound established by
Moore for his first algorithm~\cite{Moore68,Pinedo22}.

The Moore-Hodgson algorithm is not an insertion-only greedy
algorithm: a job tentatively accepted into $X$ can be removed by a
later exchange operation.  Although the algorithm is simple to
describe, its correctness rests on a careful exchange
argument.\footnote{Moore stated the supplementary algorithm without
giving a separate proof.  Subsequent proofs and treatments include
Maxwell~\cite{Maxwell70}, French~\cite{French82}, van den Akker and
Hoogeveen~\cite{AkkerHoogeveen04}, Cheriyan, Ravi, and
Skutella~\cite{CheriyanRaviSkutella21}, and
Pinedo~\cite{Pinedo22}.}

More recently, Zhao and Yuan~\cite{ZhaoYuan19} presented an entirely
different $O(n\log n)$-time algorithm for $\Prob$.  Their algorithm
builds a preemptive schedule in reverse, working backward from the
largest deadline.  In this setting the deadlines act like release
times, and the algorithm schedules an available job of minimum
positive remaining processing time.  Since there are only linearly
many job-release and job-completion events, a heap gives an
$O(n\log n)$-time implementation.

An interesting aspect of Moore's algorithm is its close relationship
to the matroid greedy algorithm.  Consider the set system
$(S,\mathcal F)$, where $\mathcal F$ is the family of feasible subsets
of $S$.  This set system is not a matroid in general, but if it were,
then Moore's algorithm would be the standard matroid greedy algorithm
for computing a minimum-weight base, where the weight of a job is its
processing time~\cite{Edmonds71,Oxley11}.  Thus, if
$(S,\mathcal F)$ were a matroid, we could immediately conclude that
Moore's algorithm computes a maximum-cardinality feasible set with
minimum total processing time.  Interestingly, Lin and Wang prove
that Moore's algorithm has precisely this property~\cite{LinWang07}.
Put differently, the matroid greedy algorithm ``works'' for
$(S,\mathcal F)$ even though this set system is not a matroid.

Given this state of affairs, two natural questions arise.  First, what
matroid-like structural properties of $(S,\mathcal F)$ allow the
matroid greedy algorithm to work correctly?  Second, is there a data
structure that can perform each feasibility query in Moore's algorithm
in $O(\log n)$ time?  Zhao and Yuan recently answered the latter
question in the affirmative in an amortized sense~\cite{ZhaoYuan21}.
Their data structure yields yet another $O(n\log n)$-time method for
solving $\Prob$.

{\bf Our contributions.} The present paper focuses on improving our
understanding of Moore's algorithm, which we refer to as ``the greedy           
algorithm'' throughout the remainder of the paper.  We make three main
contributions. First, we provide
a simple augmented BST data structure for performing each of the
feasibility queries associated with the greedy algorithm in
(worst-case) $O(\log n)$ time. The details of this data structure are
presented in Section~\ref{sec:augment}. Section~\ref{sec:zhao}
provides a high-level overview of the Zhao-Yuan data structure so that
the reader can better appreciate the relative simplicity of our
augmented BST approach.

Our second main contribution is to identify the matroid-like
structural properties of $(S,\mathcal F)$ that allow the matroid greedy
algorithm to work correctly.  We begin by arguing that if all jobs in
$S$ have the same processing time, then $(S,\mathcal F)$ is a nested
matroid---a special case of a laminar matroid in which the constraining
sets form a chain~\cite{FifeOxley17}.  Now consider the case in which
all jobs in $S$ have one of two processing times, say $p_1$ and $p_2$,
where $p_1<p_2$.  For each $i$ in $\{1,2\}$, let $S_i$ denote the set
of jobs with processing time $p_i$.  If $X$ is a
maximum-cardinality feasible subset of $S$ with minimum total
processing time, then $X\cap S_1$ is a base $B$ of the nested matroid
on $S_1$ whose independent sets are the feasible subsets of $S_1$,
and $X\cap S_2$ is a base of the nested matroid on $S_2$ whose
independent sets are the subsets $Y$ of $S_2$ such that $B\cup Y$ is
feasible.  Crucially, we prove that the latter matroid is independent
of the choice of $B$.  Thus the matroid greedy algorithm works
correctly because it first processes the jobs in $S_1$, yielding a
base of the first nested matroid, and then processes the jobs in
$S_2$, yielding a base of the second.

More generally, the feasibility structure is matroidal in a tier-by-tier
sense.  Regardless of the particular set of jobs with processing time
less than $p$ selected by the greedy algorithm, the ``addable'' (that
is, feasible to add) subsets of processing-time-$p$ jobs form the
independent sets of the same nested matroid (Lemma~\ref{lem:tierOrthog} and Theorem~\ref{thm:level}).  The problem
thus decomposes into one matroid per distinct processing time, solved
in increasing order of processing time.  Moore's scan is therefore an
ordered sequence of ordinary matroid greedy phases.  Equivalently, the
tier matroids combine into a direct-sum matroid whose bases are exactly
the complete greedy outputs.  This explains why the matroid greedy
algorithm works correctly even though $(S,\mathcal F)$ is not a
matroid.

Our third main contribution is to show that the same feasibility
structure corresponds globally to a polymatroid.  An explicit
linear-size flow network measures how much processing offered by a set
of jobs can be completed before the relevant deadlines.  Its
maximum-flow value is normalized, monotone, and submodular, and hence
is a polymatroid rank function~\cite{Edmonds70,Fujishige05}.  The
scheduling interpretation is exact: a set of jobs is feasible precisely
when all of its offered processing can be accommodated, and Moore
accepts a candidate precisely when it contributes its full processing
requirement as a marginal increase in the rank.  The usual polymatroid
greedy algorithm may allocate only part of a job, but the connection
to discrete greedy becomes exact tier by tier.  After contracting the
shorter accepted jobs and measuring residual rank in units of the
current processing time, the polymatroid recovers the complete rank
function of the corresponding nested tier matroid.  Thus the flow view
explains both Moore's all-or-nothing greedy rule and the matroid greedy
phases identified above.

In addition to the three main contributions, we strengthen some
previously known structural properties of the greedy algorithm.  For
example, over all shortest-job-first tie orders, the first $i$
accepted jobs are exactly the feasible $i$-sets with minimum work (i.e., total
processing time).\footnote{Zhao and Yuan proved that one fixed
shortest-job-first order, with ties broken by EDF, yields a nested
sequence of minimum processing time solutions, one at each attainable
cardinality~\cite{ZhaoYuan21}.  We strengthen this fixed-order result
by ranging over all shortest-job-first tie orders and proving the
converse realization: every minimum-work feasible set of an attainable
cardinality occurs as the corresponding acceptance prefix.}
Conversely, every minimum-work feasible $i$-set is realized by a
suitable tie order, and all such sets have the same processing-time
distribution.  Lin and Wang characterize the final critical sets by a
lexicographically least processing-time sequence~\cite{LinWang07}; our
result applies to every tie order and includes the converse
realization.  Ordering jobs within each processing-time tier by a
secondary weight also yields an optimum for that secondary objective
among the maximum-cardinality, minimum-work feasible sets.

The remainder of the paper is organized as follows.
Section~\ref{sec:prelims} provides formal definitions and a few basic
lemmas concerning feasibility. Section~\ref{sec:fast} presents our
augmented BST data structure, together with an overview of the
Zhao-Yuan data structure~\cite{ZhaoYuan21}.
Section~\ref{sec:props} develops several structural properties of
the greedy algorithm, including an exact characterization of the
job-sets arising as prefixes of its acceptance sequence.
Section~\ref{sec:matroid} gives our matroid interpretation of the
greedy algorithm. Section~\ref{sec:flow} develops the flow formulation
and its polymatroid consequences. Section~\ref{sec:concs} provides some
concluding remarks.

\section{Preliminaries}
\label{sec:prelims}

A job $\alpha$ has a positive real deadline $\alpha.d$, a positive
real processing time $\alpha.p$, and an integer identifier
$\alpha.i$.  A job-set $S$ is a finite set of jobs with distinct
identifiers.  For any job-set $S$, we define $\Ids{S}$ as
$\{\alpha.i\mid\alpha\in S\}$ and we define $\Proc{S}$ as
$\sum_{\alpha\in S}\alpha.p$.

We are interested in nonpreëmptive scheduling of a job-set on a
single resource starting at time zero.  A schedule for a job-set $S$
specifies the order in which to execute the jobs in $S$.  We say that
a schedule is earliest deadline first (EDF) if it lists the jobs in
nondecreasing order of deadline.  Similarly we say that a schedule is
shortest job first (SJF) if it lists the jobs in nondecreasing order
of processing time.  For any job-set $S$, we define
$\Edf{S}$ (resp., $\Sjf{S}$) as the set of all EDF (resp., SJF)
schedules of $S$.

A schedule for a job-set $S$ is said to be feasible if the completion
time of each job $\alpha$ in $S$ is at most $\alpha.d$.  A job-set
$S$ is said to be feasible if it admits a feasible schedule.  For any job-set $S$, we define $\Feasible{S}$ as the set of all feasible subsets of $S$.
For any job-set $S$ and any positive real $d$, we define
$\ProcTwo{S}{d}$ as $\sum_{\alpha\in S:\alpha.d\leq d}\alpha.p$.

We study the greedy algorithm that takes as input a job-set $S$, sorts
the jobs in $S$ in nondecreasing order of processing time to obtain a
sequence $\Seq$ in $\Sjf{S}$, and then computes a feasible subset $X$
of $S$ as follows: $X$ is initialized to the empty set; each
successive job $\alpha$ in $\Seq$ is processed by adding $\alpha$ to
$X$ if $X+\alpha$ is feasible. Section~\ref{sec:props} presents
several basic properties of the greedy algorithm, including the fact
that it computes a maximum-cardinality feasible subset of $S$.

For any sequence $\Seq$ of distinct jobs, let $\Greedy{\Seq}$ denote
the set returned by applying this insertion rule to $\Seq$. For any
job-set $S$, define
\[
\Greedy{S}
=
\{\Greedy{\Seq}\mid\Seq\in\Sjf{S}\}.
\]

We now present several useful definitions and lemmas for
characterizing feasible sets.  For any job-set $S$ and any positive
real $d$, we define $\Gap{S}{d}$ as $d-\ProcTwo{S}{d}$.

\begin{lemma}
\label{lem:feasibleGap}
Let $S$ be a job-set.  If $\Gap{S}{\alpha.d}\geq 0$ for all $\alpha$
in $S$, then $S$ is feasible, and if $S$ is feasible, then
$\Gap{S}{d}\geq 0$ for all positive reals $d$.
\end{lemma}
\begin{proof}
Observe that in any schedule in $\Edf{S}$, each job $\alpha$
completes no later than $\ProcTwo{S}{\alpha.d}$; equality holds for
the last job among those with deadline $\alpha.d$.  It follows that
if $\Gap{S}{\alpha.d}\geq 0$ for all $\alpha$ in $S$, then any
schedule in $\Edf{S}$ is feasible for $S$.

If $\Gap{S}{d}<0$ for some positive real $d$, then the jobs due by
time~$d$ require more than $d$ units of processing, so $S$ is not
feasible.  Thus if $S$ is feasible, then $\Gap{S}{d}\geq 0$ for all
positive reals $d$, as claimed.
\end{proof}

For any job-set $S$, we define $\Slack{S}$ as
$\min_{\alpha\in S}\Gap{S}{\alpha.d}$, with the convention
$\min\emptyset=+\infty$.

\begin{lemma}
\label{lem:feasibleSlack}
Let $S$ be a job-set. Then $S$ is feasible if and only if
$\Slack{S}\geq 0$.
\end{lemma}
\begin{proof}
Immediate from Lemma~\ref{lem:feasibleGap}.
\end{proof}

For any job-set $S$ and any positive real $d$, we 
define $\SlackTwo{S}{d}$ as
\[
\min(\Gap{S}{d},\min_{\alpha\in S:d\leq\alpha.d}\Gap{S}{\alpha.d}),
\]
where, as above, the minimum of the empty set is $+\infty$.
Equivalently,
\[
\SlackTwo{S}{d}=\min_{t\geq d}\Gap{S}{t}.
\]
Consequently, $\SlackTwo{S}{d}$ is nondecreasing in $d$.

\begin{lemma}
\label{lem:feasibleToAdd}
Let $X$ be a feasible job-set and let $\alpha$ be a job such that
$\alpha.i$ does not belong to $\Ids{X}$. Then $X+\alpha$ is feasible
if and only if $\SlackTwo{X}{\alpha.d}\geq\alpha.p$.
\end{lemma}\begin{proof}
Let $Y$ denote $\{\beta\in X+\alpha\mid\beta.d<\alpha.d\}$ and let $Z$
denote $(X+\alpha)\setminus Y$.  For any job $\beta$ in $Y$, we have
$\Gap{X+\alpha}{\beta.d}=\Gap{X}{\beta.d}\geq 0$ where the inequality
follows from Lemma~\ref{lem:feasibleGap} and the feasibility of $X$.
For any job $\beta$ in $Z$, we have
$\Gap{X+\alpha}{\beta.d}=\Gap{X}{\beta.d}-\alpha.p$.

For the if direction, assume
$\SlackTwo{X}{\alpha.d}\geq\alpha.p$. Then for all jobs $\beta$ in
$Z$, we have $\Gap{X}{\beta.d}\geq\alpha.p$ and hence
$\Gap{X+\alpha}{\beta.d}\geq 0$. Since $\Gap{X+\alpha}{\beta.d}\geq 0$
for all jobs $\beta$ in $Y$, we conclude that $X+\alpha$ is feasible
by Lemma~\ref{lem:feasibleGap}.

For the only if direction, assume $X+\alpha$ is
feasible. Lemma~\ref{lem:feasibleGap} implies
$\Gap{X+\alpha}{\beta.d}\geq 0$ for all $\beta$ in $Z$. Hence
$\Gap{X}{\beta.d}\geq\alpha.p$ for all $\beta$ in $Z$, implying that
$\SlackTwo{X}{\alpha.d}\geq\alpha.p$, as required.
\end{proof}

For any job-set $S$, we define $\MyBig{S}$ as the set
$\ArgMax_{X\in\Feasible{S}}|X|$ of maximum-cardinality feasible
subsets, we define $\Rank{S}$ as the cardinality of any job-set in
$\MyBig{S}$, and we define $\Short{S}$ as the set
$\ArgMin_{X\in\MyBig{S}}\Proc{X}$.  For any job-set $S$ and any
integer $i$ in $\{0,\ldots,\Rank{S}\}$, we define $\ShortTwo{S}{i}$ as
the set
\[
\ArgMin_{X\in\Feasible{S}:|X|=i}\Proc{X}.
\]
For any job-set $S$ and any positive real $p$, we define
$\AtMost{S}{p}$ as $\{\alpha\in S\mid \alpha.p\leq p\}$ and we define
the expressions $\EqualTo{S}{p}$ and $\LessThan{S}{p}$ analogously.
For any two job-sets $X$ and $Y$, we define the expression
$X\SameDist Y$ to mean that $X$ and $Y$ have the same distribution of
processing times (i.e., $|\EqualTo{X}{p}|=|\EqualTo{Y}{p}|$ for all
positive reals $p$).

\section{Fast Implementation}
\label{sec:fast}
In this section, we develop an efficient direct implementation of the
greedy algorithm.  As in the recent work of Zhao and
Yuan~\cite{ZhaoYuan21}, we seek a fast data structure for performing
each of the $n$ feasibility queries associated with the greedy
algorithm.  In Section~\ref{sec:augment}, we present a simple
augmented BST data structure that processes each feasibility query in
$O(\log n)$ time, thereby yielding an $O(n\log n)$-time implementation
of the greedy algorithm. The only previous data structure for
performing these feasibility queries in sublinear time is that of Zhao
and Yuan~\cite{ZhaoYuan21}, which runs in amortized $O(\log n)$ time,
and hence also yields an $O(n\log n)$-time implementation of the
greedy algorithm.  To allow the reader to appreciate the relative
simplicity of our augmented BST data structure (in terms of both ease
of description and ease of establishing correctness),
Section~\ref{sec:zhao} provides a high-level overview of the main
ideas underlying the Zhao-Yuan data structure.

\subsection{An Augmented BST Data Structure}
\label{sec:augment}

In this section we describe an augmented BST data structure $T$ that
we maintain alongside the current feasible job-set $X$ in our
efficient implementation of the greedy algorithm.

Tree $T$ supports two operations $\TreeAdd{T}{d}{p}$ and
$\TreeSlack{T}$ that run in logarithmic time and constant time,
respectively.\footnote{See Cormen et al.~\cite[Ch.~14]{CLRS09} for an
introduction to augmented binary search trees. As in this
introduction, we can assume that the underlying BST is a red-black
tree.} Before defining these two operations, we describe some basic
characteristics of the tree $T$ along with the invariants relating the
state of $T$ with that of $X$.

Each node $u$ of $T$ has two primary fields $u.d$ and $u.p$ augmented
by two auxiliary fields $u.x$ and $u.y$. The $d$-values are all
distinct and serve as the keys of BST $T$: an inorder traversal of $T$
visits the nodes in increasing order of $d$-value. For any subtree
$T'$ of $T$, we define $\Keys{T'}$ as $\{u.d\mid u\in T'\}$.  For any
node $u$ in $T$, we define $T_u$ as the subtree of $T$ rooted at node
$u$, $u.\ell$ (resp., $u.r$) as the left (resp., right) child of $u$
(or $\Nil$ if there is no such child), and $X_u$ as $\{\alpha\in
X\mid\alpha.d\in\Keys{T_u}\}$. We define the expressions $\Nil.x$ and
$\Nil.y$ as $0$ and $+\infty$, respectively. We maintain the following
invariants relating $X$ and $T$: $\Keys{T}=\{\alpha.d\mid\alpha\in
X\}$; for any node $u$ in $T$, $u.p=\sum_{\alpha\in
  X:\alpha.d=u.d}\alpha.p$, $u.x=\Proc{X_u}$, and $u.y=\Slack{X_u}$.

The operation $\TreeSlack{T}$ returns the $y$-value of the root of
$T$, and hence runs in constant time.  The operation
$\TreeAdd{T}{d}{p}$ takes a positive real $d$ and a nonzero real $p$
and modifies the primary fields of tree $T$ as follows. If $T$
contains a node $u$ with $u.d=d$, the effect of this operation is to
add $p$ to $u.p$ and to remove node $u$ from $T$ if the resulting
value of $u.p$ is zero. Otherwise, the effect of the operation is to
introduce a node $u$ with $u.d=d$ and $u.p=p$. It is straightforward
to obtain an $O(\log |T|)$-time red-black tree implementation of the
$\TreeAdd{T}{d}{p}$ operation (e.g., by using a red-black tree find
possibly followed by a red-black tree insertion or deletion).
We consider only calls for which the resulting processing mass at key
$d$ is nonnegative; if $T$ contains no node with key $d$, we require
$p>0$.

We now establish that the auxiliary fields can be efficiently
maintained during a $\TreeAdd{T}{d}{p}$ operation. Following the
framework of Cormen et al.~\cite[Ch.~14]{CLRS09}, it is sufficient to
argue that the values of the auxiliary fields of a node $u$ can be
determined in constant time given the values of the primary fields of
$u$ and the values of all fields of the children of $u$. To see this,
observe that $u.x=u.p+u.\ell.x+u.r.x$ and
$u.y=\min(u.\ell.y,\min(u.d,u.r.y)-u.\ell.x-u.p)$.

When we want to check whether a given job $\alpha$ can be added to the
current feasible set $X$ to obtain a larger feasible set $X+\alpha$,
we invoke $\TreeAdd{T}{\alpha.d}{\alpha.p}$ to tentatively update $T$
to reflect the set $X+\alpha$. Based on our invariant and Lemma~\ref{lem:feasibleSlack},
we can determine whether $X+\alpha$ is feasible by checking whether
$\TreeSlack{T}$ returns a nonnegative value.  If not, we use a
$\TreeAdd{T}{\alpha.d}{-\alpha.p}$ operation to restore the invariants
with respect to the feasible set $X$ (as opposed to $X+\alpha$).

In Appendix~\ref{sec:variant}, we describe a variant of data structure
$T$ that is slightly more complicated to describe but has the
advantage that the tree is only modified when a job $\alpha$ is added
to the current feasible set $X$.

\subsection{Overview of the Zhao-Yuan Data Structure}
\label{sec:zhao}
The purpose of the present section is to highlight the main ideas
underlying the Zhao-Yuan data structure. Our overview is not intended
to be comprehensive; see~\cite{ZhaoYuan21} for full details. Our goal
is to allow the reader to appreciate the relative simplicity of the
augmented BST data structure presented in Section~\ref{sec:augment}.

At a high level, the Zhao-Yuan data structure is based on three main
ideas. The first idea is to exploit the fact that a job-set admits a
feasible nonpreëmptive schedule if and only if it admits a feasible
preëmptive schedule. In other words, the Zhao-Yuan data structure
determines whether $X+\alpha$ is feasible by determining whether there
is a feasible preëmptive schedule for $X+\alpha$.

The second idea is to observe that, in order to construct a feasible
preëmptive schedule for a given feasible job-set, we can schedule
each successive job $\alpha$ as late as possible; that is, we allocate
unused intervals working backwards from the deadline of $\alpha$ until
the processing time of $\alpha$ has been fulfilled.

The above two ideas are sufficient to provide the basis for a correct
data structure. We can store the maximal contiguous allocated
intervals associated with the current schedule in a BST ordered by the
interval start times. To check whether we can add a job $\alpha$ with
deadline $d$ and processing time $p$, we first search the BST to
identify the rightmost allocated interval with start time at most $d$.
We then move left through the interval records while maintaining a
running sum of the unused time available before $d$, until either
(1) the sum reaches $p$, so $\alpha$ can be added while preserving
feasibility, or (2) time zero is reached and the sum remains less than
$p$, so $\alpha$ cannot be added while preserving feasibility.

A single test may encounter many interval records. If each movement
between records is implemented by an ordered-dictionary operation, the
immediate bound is $O(k\log n)$, where $k$ is the number of records
encountered before one of the two terminating conditions is reached.
When condition~(1) is reached, however, inserting $\alpha$ as late as
possible transforms the portion of the schedule encountered by the
search. The restoration procedures consolidate that portion, so all but
a constant number of the records crossed by the search disappear as
separate records. Thus the work spent crossing those records can be
charged against the resulting decrease in the size of the
representation. This gives the desired
amortized logarithmic bound for successful tests. The same conclusion
does not yet follow for a test ending in condition~(2), because a
rejected job need not alter the representation. The third idea removes
this difficulty.

The third main idea is to observe that if a job $\alpha$ with deadline
$d$ and processing time $p$ cannot be added to the current feasible set,
then no job that is processed by the greedy algorithm after
$\alpha$ and that has deadline at most $d$ will ever be accepted.
Indeed, such a job has processing time at least $p$ and deadline at
most $d$; replacing it by $\alpha$ in any feasible schedule would make
$X+\alpha$ feasible, a contradiction. We may therefore modify the preëmptive schedule on the prefix relevant
to deadline $d$: rather than continuing to execute the jobs represented
there as late as possible, we schedule them from time zero as early as
possible, in EDF order. The affected interval records can then be
consolidated into a single initial occupied interval. Consequently, a
rejected test also pays for the records encountered by its search, and
the same amortized analysis applies to both outcomes.

The kind of schedule produced by these operations is called a
DD-schedule. An initial portion of the accepted work is packed
from time zero in EDF order; the remaining work is scheduled strictly
after this initial portion and as late as possible. Equivalently, the
first occupied interval begins at time zero, while each later occupied
interval is anchored at a deadline (that is, its right endpoint is a
deadline), with idle gaps separating successive occupied intervals.
The ordered family of maximal occupied intervals, and hence the idle
gaps determined by their endpoints, is referred to as the associated
DD-interval family. The implementation need not retain the full
assignment of job fragments to time intervals: it is enough to maintain
this interval family and the information needed to perform the backward
search and restore DD form.

The process of determining whether to accept a given candidate job
requires traversing a (potentially large) number of intervals, and either outcome
transforms the portion crossed by the search. Zhao and Yuan bound the
total number of ordered-dictionary operations over the complete scan by
a potential argument, obtaining $O(\log n)$ amortized time per
candidate and $O(n\log n)$ total time; see~\cite[Lem.~2.4 and
Theorem~2.1]{ZhaoYuan21}.

The Zhao-Yuan structure and our augmented BST implement the same
accept-if-feasible rule and lead to the same overall asymptotic running
time. The difference lies in the representation and analysis. The
Zhao-Yuan structure maintains a preemptive schedule through occupied
intervals and idle gaps, and its time bound is amortized over the
resulting interval transformations. Our tree data structure $T$ does not
maintain a schedule: for each accepted deadline $d$, $T$ contains a node
$u$ with key $u.d=d$ and associated processing mass $u.p$, as well as two
auxiliary fields $u.x$ and $u.y$ determined by subtree $T_u$. Moreover, the
variant described in Appendix~\ref{sec:variant} only modifies $T$ when a job
is accepted.

\section{Properties of the Greedy Algorithm}
\label{sec:props}

We begin with an exchange observation. It gives a short proof of the
known minimum-work property of Moore's greedy algorithm and also
characterizes every prefix of the jobs accepted by the algorithm.

\medskip
\noindent\textbf{Exchange observation.}
Let $S$ be a job-set, let $\Seq$ belong to $\Sjf{S}$, let $\Seq'$
denote a prefix of $\Seq$, and put $X=\Greedy{\Seq'}$. Let $Y$ belong
to $\ShortTwo{S}{i}$ for some $i\geq |X|$. If $X\nsubseteq Y$, then
there exist jobs $\alpha^*$ in $X\setminus Y$ and $\beta^*$ in
$Y\setminus X$ such that $\alpha^*.p=\beta^*.p$ and
$Y+\alpha^*-\beta^*$ belongs to $\ShortTwo{S}{i}$.

\begin{proof}
Let $\alpha^*$ be the first job in $\Seq'$ that belongs to $X\oplus
Y$. Every earlier job belongs either to both $X$ and $Y$ or to neither.
If $\alpha^*$ belonged to $Y\setminus X$, then the set held by the
greedy algorithm when $\alpha^*$ was considered would be a subset of
the feasible set $Y$. Thus $\alpha^*$ would have been accepted, a
contradiction. Hence $\alpha^*$ belongs to $X\setminus Y$. Since $|Y|\geq|X|$, the set
$Y\setminus X$ is nonempty; let $\beta^*$ be one of its
minimum-deadline jobs. Since $\beta^*$ belongs to $Y\setminus X$, the choice of $\alpha^*$
implies that $\alpha^*$ precedes $\beta^*$ in $\Seq$,
and hence $\alpha^*.p\leq\beta^*.p$.

Put $Y'=Y+\alpha^*-\beta^*$. If a job $\alpha$ in $Y'$ satisfies
$\alpha.d\geq\beta^*.d$, then
\[
\ProcTwo{Y'}{\alpha.d}
\leq \ProcTwo{Y}{\alpha.d}-\beta^*.p+\alpha^*.p
\leq \alpha.d.
\]
If $\alpha.d<\beta^*.d$, then every job in $Y$ with deadline less than
$\beta^*.d$ belongs to $X$, by the choice of $\beta^*$; hence
\[
\ProcTwo{Y'}{\alpha.d}\leq\ProcTwo{X}{\alpha.d}\leq\alpha.d.
\]
Thus $Y'$ is feasible by Lemma~\ref{lem:feasibleGap}. 
Since $Y$ has minimum work among the feasible $i$-sets, we must have
$\Proc{Y'}=\Proc{Y}$. Hence $\alpha^*.p=\beta^*.p$, and $Y'$ belongs
to $\ShortTwo{S}{i}$.

\end{proof}

\begin{lemma}
\label{lem:greedyIsShort}
Let $S$ be a job-set. Then $\Greedy{S}$ is contained in $\Short{S}$.
\end{lemma}
\begin{proof}
Let $X$ belong to $\Greedy{S}$ and $Y$ to $\Short{S}$. Repeatedly
applying the exchange observation gives a job-set $Y'$ in $\Short{S}$
such that $X\subseteq Y'$ and $Y\SameDist Y'$. Every greedy output is
inclusion-maximal in $\Feasible{S}$: otherwise, the first job of a
proper feasible extension not already accepted would have been
accepted when it was considered. Hence $X=Y'$, so $X$ belongs to
$\Short{S}$.
\end{proof}

Lemma~\ref{lem:greedyIsShort} implies every job-set in $\Greedy{S}$
has cardinality $\Rank{S}$. For any job-set $S$, any sequence $\Seq$
in $\Sjf{S}$, and any $i$ in $\{0,\ldots,\Rank{S}\}$, we define
$\GreedyTwo{\Seq}{i}$ as the job-set consisting of the first $i$ jobs
accepted by the greedy algorithm on input $\Seq$. For any job-set $S$
and any such $i$, we define
\[
\GreedyTwo{S}{i}
=
\{\GreedyTwo{\Seq}{i}\mid\Seq\in\Sjf{S}\}.
\]
Thus $\GreedyTwo{S}{\Rank{S}}=\Greedy{S}$.

\begin{lemma}
\label{lem:prefix}
Let $S$ be a job-set and let $i$ belong to
$\{0,\ldots,\Rank{S}\}$. Then
\[
\GreedyTwo{S}{i}=\ShortTwo{S}{i}.
\]
Moreover, $X\SameDist Y$ for all job-sets $X$ and $Y$ in
$\GreedyTwo{S}{i}$.
\end{lemma}
\begin{proof}
Let $X$ belong to $\GreedyTwo{S}{i}$ and $Y$ to $\ShortTwo{S}{i}$.
Repeatedly applying the exchange observation to the prefix producing
$X$ gives a job-set $Y'$ in $\ShortTwo{S}{i}$ such that $X\subseteq
Y'$ and $Y\SameDist Y'$. Since $|X|=|Y'|=i$, we have $X=Y'$. Thus
$X$ belongs to $\ShortTwo{S}{i}$ and $X\SameDist Y$.

Conversely, let $Y$ belong to $\ShortTwo{S}{i}$, and choose a sequence
$\Seq$ in $\Sjf{S}$ that breaks ties in favor of jobs in $Y$. Put
$X=\GreedyTwo{\Seq}{i}$. By the preceding paragraph, $X$ and $Y$ have
the same processing-time distribution. If $X\neq Y$, let $p$ be the
least processing time at which they differ. Since $X$ and $Y$ have the
same processing-time distribution, $\EqualTo{X}{p}$ and
$\EqualTo{Y}{p}$ have the same positive cardinality. Thus fewer than
$i$ jobs have been accepted before this tier, and the accepted job-set
at that point is $\LessThan{X}{p}=\LessThan{Y}{p}$. Every job in
$\EqualTo{Y}{p}$ is therefore accepted, since the job-set held when it
is considered is a subset of the feasible set $Y$. As these jobs are
considered before the remaining jobs of processing time $p$, this
contradicts $\EqualTo{X}{p}\neq\EqualTo{Y}{p}$. Hence $X=Y$, and so
$Y$ belongs to $\GreedyTwo{S}{i}$.
\end{proof}

One fixed SJF/EDF order was previously known to produce nested
minimum-work solutions~\cite{ZhaoYuan21}; Lin and Wang characterize the
final critical sets by a lexicographically least processing
sequence~\cite{LinWang07}. Lemma~\ref{lem:prefix} applies to every tie
order and also gives the converse realization.

For any job-set $S$, let $\Lex{S}$ be the set of all job-sets $X$ in
$\Feasible{S}$ such that
\[
|\AtMost{X}{p}|\geq|\AtMost{Y}{p}|
\]
for every $Y$ in $\Feasible{S}$ and every positive real $p$. As an
immediate corollary,
\[
\Greedy{S}=\Lex{S}.
\]
Indeed, if $X$ belongs to $\Greedy{S}$, then $\AtMost{X}{p}$ belongs
to $\Greedy{\AtMost{S}{p}}$ for every $p$, so
Lemma~\ref{lem:greedyIsShort} gives the defining inequalities for
$\Lex{S}$. Conversely, if $X$ belongs to $\Lex{S}$, then $X$ has
maximum cardinality and its sorted processing-time vector is
componentwise no larger than that of any maximum-cardinality feasible
job-set. Hence $X$ belongs to $\Short{S}$, and
Lemma~\ref{lem:prefix} gives $X\in\Greedy{S}$.

\subsection{Tier Orthogonality}
\label{sec:orthog}

For any job-set $S$ and any positive real $p$, define
\[
\Bases_{S,p}
=
\{\EqualTo{\Greedy{\Seq}}{p}\mid\Seq\in\Sjf{S}\}.
\]

\begin{lemma}
\label{lem:tierOrthog}
Let $S$ be a job-set. Then
\[
\Greedy{S}
=
\{X\subseteq S\mid
  \EqualTo{X}{p}\in\Bases_{S,p}\text{ for every positive real }p\}.
\]
\end{lemma}
\begin{proof}
We first note that if every job in a job-set $R$ has processing time at
most that of a job $\alpha$, then, for every $\Seq$ in $\Sjf{R}$,
\[
\alpha\in\Greedy{\Seq\alpha}
\quad\Longleftrightarrow\quad
\Rank{R+\alpha}=\Rank{R}+1,
\]
by Lemma~\ref{lem:greedyIsShort}. Thus the acceptance of $\alpha$ is
independent of the SJF order chosen for $R$. Applying this observation
successively shows that, for any fixed order of the jobs of processing
time $p$, the selected jobs of processing time $p$ are independent of
the SJF order chosen for the shorter jobs.

Now prescribe a member $B_p$ of $\Bases_{S,p}$ for every $p$. For each
$p$, choose from a greedy run realizing $B_p$ the induced order of
$\EqualTo{S}{p}$, and concatenate these tier orders in increasing
order of $p$. The preceding observation shows successively that the
greedy algorithm selects exactly $B_p$ from each tier. This proves the
nontrivial inclusion; the reverse inclusion follows directly from the
definition of $\Bases_{S,p}$.
\end{proof}

\section{Matroid Interpretation}
\label{sec:matroid}

For any set $U$ and any subset $\Ind$ of $2^U$, the pair $(U,\Ind)$ is
a matroid if the following three conditions are satisfied: (1) $\Ind$
is nonempty; (2) for any $X$ in $\Ind$ and any subset $Y$ of $X$, $Y$
belongs to $\Ind$ (the hereditary property); (3) for any $X$ and $Y$
in $\Ind$ such that $|X|<|Y|$, there exists $y$ in $Y\setminus X$ such
that $X+y$ belongs to $\Ind$ (the exchange property).

For any job-set $S$ and any positive real $p$, we define
$\Ind_{S,p}$ as
\[
\{X\subseteq\EqualTo{S}{p}\mid\exists Y\in\Bases_{S,p}:X\subseteq Y\}.
\]

\begin{lemma}
\label{lem:tierMatroid}
Let $S$ be a job-set and let $p$ be a positive real. Then
$(\EqualTo{S}{p},\Ind_{S,p})$ is a matroid.
\end{lemma}
\begin{proof}
Since $\emptyset$ belongs to $\Ind_{S,p}$, the family $\Ind_{S,p}$ is
nonempty. The hereditary property follows immediately from its
definition. It remains to prove the exchange property.

Let $X$ and $Y$ belong to $\Ind_{S,p}$, where $|X|<|Y|$. By the
definition of $\Ind_{S,p}$, some member of $\Bases_{S,p}$ contains
$X$. Restricting an SJF schedule that realizes this member to
$\LessThan{S}{p}\cup X$ yields an SJF schedule $\Seq$ for which
$\EqualTo{\Greedy{\Seq}}{p}=X$.
Let $\Seq'$ be an SJF schedule of $Y\setminus X$ and let
$Z$ denote $\EqualTo{\Greedy{\Seq\Seq'}}{p}$.
Then
$X\subseteq Z\subseteq X\cup Y$.
Extending $\Seq\Seq'$ to an SJF schedule of $S$ shows that $Z$ is
contained in a member of $\Bases_{S,p}$; hence $Z$ belongs to
$\Ind_{S,p}$.

It remains to show that $Z$ properly contains $X$. Since $Y$ belongs
to $\Ind_{S,p}$, there is an SJF schedule of
$\LessThan{S}{p}\cup Y$ in which every job of $Y$ is accepted.
Appending the jobs of $X\setminus Y$ gives an SJF schedule of
$R=\LessThan{S}{p}\cup X\cup Y$
whose greedy output contains at least $|Y|$ jobs of processing time
$p$. By Lemma~\ref{lem:greedyIsShort}, all greedy runs on $R$ accept
the same total number of jobs. Their restrictions to
$\LessThan{S}{p}$ likewise accept the same number of jobs, again by
Lemma~\ref{lem:greedyIsShort}. Consequently,
$|Z|\geq |Y|>|X|$.
Thus there is a job $y$ in
$Z\setminus X\subseteq Y\setminus X$.
Since $X+y\subseteq Z$ and $Z$ belongs to $\Ind_{S,p}$, the hereditary
property already established implies $X+y$ belongs to
$\Ind_{S,p}$. This proves the exchange property.
\end{proof}

For any job-set $S$ and any positive real $p$, we define
$\Matroid_{S,p}$ as the matroid $(\EqualTo{S}{p},\Ind_{S,p})$ of
Lemma~\ref{lem:tierMatroid}. The definition of $\Ind_{S,p}$ implies
$\Bases_{S,p}$ is the set of all bases of matroid $\Matroid_{S,p}$.

For any job-set $S$, we define $\Ind_S$ as
\[
\{X\subseteq S\mid\forall p\in\PosReals:\EqualTo{X}{p}\in\Ind_{S,p}\}.
\]

\begin{lemma}
\label{lem:matroidOverall}
Let $S$ be a job-set. Then $\Ind_S$ is contained in
$\Feasible{S}$ and $(S,\Ind_S)$ is a matroid.
\end{lemma}
\begin{proof}
Let $X$ belong to $\Ind_S$. For each processing time $p$, extend
$\EqualTo{X}{p}$ to a basis $B_p$ of $\Matroid_{S,p}$. By the
definition of $\Ind_{S,p}$, each $B_p$ belongs to $\Bases_{S,p}$.
Lemma~\ref{lem:tierOrthog} therefore implies
$B=\bigcup_p B_p$ belongs to $\Greedy{S}$ and hence is feasible.
Since $X\subseteq B$, the hereditary property of feasibility implies
$X$ is feasible. Thus $\Ind_S$ is contained in $\Feasible{S}$.

We now argue that $(S,\Ind_S)$ is a matroid.  Since $\emptyset$
belongs to $\Ind_S$, we know that $\Ind_S$ is nonempty.  The
hereditary property is immediate. It remains to establish the exchange
property.  Let $X$ and $Y$ belong to $\Ind_S$ and assume that
$|X|<|Y|$. We need to prove that there exists a $y$ in $Y\setminus X$
such that $X+y$ belongs to $\Ind_S$. Since $|Y|>|X|$, there is a
positive real $p$ such that $|\EqualTo{X}{p}|<|\EqualTo{Y}{p}|$. Since
$M_{S,p}=(\EqualTo{S}{p},\Ind_{S,p})$ is a matroid, there exists a $y$
in $\EqualTo{Y}{p}\setminus\EqualTo{X}{p}$ such that
$\EqualTo{X}{p}+y$ belongs to $\Ind_{S,p}$. Observe that $y$ belongs
to $Y\setminus X$. Moreover, the definition of $\Ind_S$ implies $X+y$
belongs to $\Ind_S$, as required.
\end{proof}

For any job-set $S$, we define $\Matroid_S$ as the matroid of
Lemma~\ref{lem:matroidOverall} and we define $\Bases_S$ as the set of
all bases of $\Matroid_S$.

\begin{lemma}
\label{lem:basesGreedy}
Let $S$ be a job-set. Then $\Bases_S=\Greedy{S}$.
\end{lemma}
\begin{proof}
A basis of a direct sum is the union of one basis from each
summand~\cite[Sec.~4.2]{Oxley11}; the result now follows from
Lemma~\ref{lem:tierOrthog} and the definition of $\Ind_S$.
\end{proof}

\subsection{The Structure of Matroid $\Matroid_{S,p}$}
\label{sec:structure}

For any feasible job-set $X$ and any job $\alpha$ such that $\alpha.i$
does not belong to $\Ids{X}$, we define $\Level{X}{\alpha}$ as
$\lfloor\SlackTwo{X}{\alpha.d}/\alpha.p\rfloor$.

\begin{lemma}
\label{lem:addLevel}
For any feasible job-set $X$ and any job $\alpha$ such that
$\alpha.i$ does not belong to $\Ids{X}$, the job-set $X+\alpha$ is
feasible if and only if $\Level{X}{\alpha}\geq 1$.
\end{lemma}
\begin{proof}
Immediate from Lemma~\ref{lem:feasibleToAdd}.
\end{proof}

\begin{lemma}
\label{lem:addLevelTwo}
Let $X$ be a feasible job-set, let $\alpha$ and $\beta$ be jobs such
that $\alpha.i\not=\beta.i$ and
$\{\alpha.i,\beta.i\}\cap\Ids{X}=\emptyset$, and assume that
$\alpha.d\leq\beta.d$, $\alpha.p=\beta.p$, and $\Level{X}{\alpha}\geq
1$.  Then job-set $X+\alpha$ is feasible and
$\Level{X+\alpha}{\beta}=\Level{X}{\beta}-1$.
\end{lemma}
\begin{proof}
Lemma~\ref{lem:addLevel} implies $X+\alpha$ is feasible.  To establish
that $\Level{X+\alpha}{\beta}=\Level{X}{\beta}-1$, it is sufficient to
prove that $\SlackTwo{X+\alpha}{\beta.d}=\SlackTwo{X}{\beta.d}-\beta.p$;
indeed, for every real $z$ and positive $p$,
$\lfloor(z-p)/p\rfloor=\lfloor z/p\rfloor-1$.
We have
\begin{eqnarray*}
\SlackTwo{X+\alpha}{\beta.d}
& = &
\min_{\gamma\in X+\alpha+\beta:\beta.d\leq\gamma.d}\Gap{X+\alpha}{\gamma.d}\\
& = &
\min_{\gamma\in X+\beta:\beta.d\leq\gamma.d}\Gap{X+\alpha}{\gamma.d}\\
& = &
\min_{\gamma\in X+\beta:\beta.d\leq\gamma.d}\Gap{X}{\gamma.d}-\alpha.p\\
& = &
\SlackTwo{X}{\beta.d}-\beta.p,
\end{eqnarray*}
where the second equality holds immediately if $\alpha.d<\beta.d$ and
also holds if $\alpha.d=\beta.d$ since $\alpha.d=\beta.d$ implies
$\Gap{X+\alpha}{\alpha.d}=\Gap{X+\alpha}{\beta.d}$.
\end{proof}

\begin{lemma}
\label{lem:addLevelThree}
Let $X$ be a feasible job-set and let
$Y=\{\alpha_1,\ldots,\alpha_k\}$ be a job-set such that $\Ids{X}\cap
\Ids{Y}=\emptyset$, $\alpha_1.d\leq\cdots\leq\alpha_k.d$, and
$\alpha_1.p=\cdots=\alpha_k.p$. Then job-set $X\cup Y$ is feasible if
and only if $\Level{X}{\alpha_i}\geq i$ for all $i$ in
$\{1,\ldots,k\}$.
\end{lemma}
\begin{proof}
For $i$ in $\{0,\ldots,k\}$, let
$Y_i=\{\alpha_1,\ldots,\alpha_i\}$. Whenever $X\cup Y_{i-1}$ is
feasible, repeated application of Lemma~\ref{lem:addLevelTwo} gives
\[
\Level{X\cup Y_{i-1}}{\alpha_i}
=\Level{X}{\alpha_i}-(i-1).
\]
Lemma~\ref{lem:addLevel} therefore implies, inductively on $i$, that
$X\cup Y_i$ is feasible for every $i$ if and only if
$\Level{X}{\alpha_i}\geq i$ for every $i$. Since $Y_k=Y$, the result
follows.
\end{proof}

Given sets $U_1\subseteq\cdots\subseteq U_k\subseteq U$ and integers
$0\leq n_1\leq\cdots\leq n_k$, it is known that $(U,\Ind)$ is a matroid if
we define $\Ind$ as the set of all subsets $X$ of $U$ such that
$|X\cap U_i|\leq n_i$ for all $i$ in $\{1,\ldots,k\}$. Such a matroid
(which is a special case of a laminar matroid) is said to be nested~\cite{FifeOxley17}.

\begin{theorem}
\label{thm:level}
For any job-set $S$ and any positive real $p$, the matroid
$\Matroid_{S,p}$ is nested.
\end{theorem}
\begin{proof}
Let $X$ belong to $\Greedy{\LessThan{S}{p}}$, let $A$ denote $\EqualTo{S}{p}$, and let
$L_i$ denote $\{\alpha\in A\mid\Level{X}{\alpha}\leq i\}$
for all $i$ in $\{0,\ldots,|A|\}$.
Let $Y=\{\alpha_1,\ldots,\alpha_k\}$ be a subset of $A$ with
$\alpha_1.d\leq\cdots\leq\alpha_k.d$. 

We first observe that
\[
Y\in\Ind_{S,p}
\quad\Longleftrightarrow\quad
X\cup Y\text{ is feasible}.
\]
For the forward implication, extend $Y$ to a member of
$\Bases_{S,p}$ and apply Lemma~\ref{lem:tierOrthog}. Conversely, if
$X\cup Y$ is feasible, use an SJF order realizing $X$ on the shorter
jobs and place the jobs of $Y$ first within tier $p$. Every job of $Y$
is then accepted, so $Y$ is contained in a member of
$\Bases_{S,p}$.

Since level is nondecreasing
with deadline, Lemma~\ref{lem:addLevelThree} gives
\[
\begin{aligned}
Y\in\Ind_{S,p}
&\iff X\cup Y\text{ is feasible}\\
&\iff \Level{X}{\alpha_j}\geq j\mbox{ for all $j$ in $\{1,\ldots,k\}$}\\
&\iff |Y\cap L_i|\leq i\mbox{ for all $i$ in $\{0,\ldots,|A|\}$}.
\end{aligned}
\]
Since $L_0\subseteq L_1\subseteq\cdots\subseteq L_{|A|}$, these are
nested-matroid constraints.
\end{proof}

As a further consequence, if ties within each processing-time tier
are broken by nondecreasing secondary weight, then the greedy output
has minimum total secondary weight among all members of $\Greedy{S}$.
Indeed, within each tier the algorithm is the standard
minimum-weight-base greedy algorithm for $\Matroid_{S,p}$, and the
tier matroids form a direct sum.

\section{A Flow-Polymatroid View}
\label{sec:flow}
The augmented BST of Section~\ref{sec:augment} answers one feasibility query at a time.  A flow network
packages all deadline prefixes into one rank function.  The resulting
threshold formula has three complementary consequences: it characterizes
feasibility by full rank contribution, defines a scheduling polymatroid for
fractional capacity, and, after contraction and quantization, recovers the
nested matroid governing each equal-processing-time tier.

\subsection{A Threshold-Cut Rank}

Let
\[
  0=d_0<d_1<\cdots<d_m
\]
be the distinct deadlines in $S$, put $\Delta_k=d_k-d_{k-1}$, and choose
$C>\Proc{S}$.  For each job $\alpha$, let $j(\alpha)$ be the index such that
$\alpha.d=d_{j(\alpha)}$.  For $X\subseteq S$, construct $\mathcal N(X)$ with
a job vertex $v_\alpha$ for each $\alpha\in X$, an interval vertex $u_k$ for
each $1\leq k\leq m$, and arcs
\[
\begin{array}{rcll}
  s&\to&v_\alpha&(\alpha.p),\\
  v_\alpha&\to&u_{j(\alpha)}&(C),\\
  u_k&\to&u_{k-1}&(C)\qquad(2\leq k\leq m),\\
  u_k&\to&t&(\Delta_k)\qquad(1\leq k\leq m),
\end{array}
\]
where parentheses give capacities.  Flow entering $u_{j(\alpha)}$ may use
that interval or move backward to any earlier interval.  The network has
$O(|X|+m)$ vertices and arcs.  Let $r(X)$ denote its maximum-flow value.
Equivalently, one may build the network for all of $S$ and delete the source
arcs of jobs outside $X$.

\begin{theorem}[Threshold-cut formula]
\label{thm:threshold-cut}
For every $X\subseteq S$,
\begin{equation}
  r(X)=\min_{0\leq\ell\leq m}
       \left(d_\ell+
       \sum_{\substack{\alpha\in X\\\alpha.d>d_\ell}}\alpha.p\right).
  \label{eq:flow-rank}
\end{equation}
\end{theorem}
\begin{proof}
The cut whose source side is $\{s\}$ has capacity $\Proc{X}<C$; hence a
minimum cut crosses no $C$-arc.  Its source-side interval vertices form a
prefix, because $u_k$ on the source side forces $u_{k-1}$ there as well.
Write this prefix as $u_1,\ldots,u_\ell$, allowing $\ell=0$.

Fix the interval prefix.  When $\alpha.d\leq d_\ell$, placing $v_\alpha$ on
the source side crosses neither $s\to v_\alpha$ nor a $C$-arc.  When
$\alpha.d>d_\ell$, the vertex $u_{j(\alpha)}$ is on the sink side, so
$v_\alpha$ must be there as well.  The minimum cut with this interval prefix
therefore has capacity
\[
  \sum_{k\leq\ell}\Delta_k+
  \sum_{\alpha\in X:\,\alpha.d>d_\ell}\alpha.p
  =
  d_\ell+
  \sum_{\alpha\in X:\,\alpha.d>d_\ell}\alpha.p.
\]
For every $\ell$, these placements define a cut with the displayed capacity.
The max-flow min-cut theorem now proves~\eqref{eq:flow-rank}.
\end{proof}

Figure~\ref{fig:flow-threshold} shows the network and one threshold cut.

\begin{figure}[b]
\centering
\begin{tikzpicture}[
  >=Stealth,
  v/.style={draw,circle,minimum size=7mm,inner sep=1pt,font=\footnotesize},
  lab/.style={font=\footnotesize}
]
  \node[v] (s) at (0,0) {$s$};
  \node[v] (a) at (1.8,1.5) {$\alpha$};
  \node[v] (b) at (1.8,0) {$\beta$};
  \node[v] (c) at (1.8,-1.5) {$\gamma$};
  \node[v] (u1) at (4.3,1.5) {$u_1$};
  \node[v] (u2) at (4.3,0) {$u_2$};
  \node[v] (u3) at (4.3,-1.5) {$u_3$};
  \node[v] (t) at (6.8,0) {$t$};

  \draw[->] (s)--node[above left,lab] {$4$}(a);
  \draw[->] (s)--node[above,lab] {$1$}(b);
  \draw[->] (s)--node[below left,lab] {$1$}(c);

  \draw[->,densely dotted] (a)--node[above,lab] {$C$}(u1);
  \draw[->,densely dotted] (b)--node[above,lab] {$C$}(u2);
  \draw[->,densely dotted] (c)--node[below,lab] {$C$}(u3);
  \draw[->,densely dotted] (u3)--node[right,lab] {$C$}(u2);
  \draw[->,densely dotted] (u2)--node[right,lab] {$C$}(u1);

  \draw[->] (u1)--node[above,lab] {$3$}(t);
  \draw[->] (u2)--node[above,lab] {$2$}(t);
  \draw[->] (u3)--node[below,lab] {$2$}(t);

  \node[draw,dashed,circle,fit=(s),inner sep=2.5pt] {};
  \node[draw,dashed,rounded corners,fit=(a)(u1),inner sep=4pt,
        label={[lab]above:source-side vertices for $d_1=3$}] {};
  \node[lab,anchor=west] at ($(t)+(7mm,10mm)$)
    {$(\alpha.p,\alpha.d)=(4,3)$};
  \node[lab,anchor=west] at ($(t)+(7mm,2mm)$)
    {$(\beta.p,\beta.d)=(1,5)$};
  \node[lab,anchor=west] at ($(t)+(7mm,-6mm)$)
    {$(\gamma.p,\gamma.d)=(1,7)$};
\end{tikzpicture}
\caption{The flow network and the threshold cut at $d_1=3$.  Its capacity is
$3+1+1=5$: the first interval contributes three units, and the two jobs due
later than $3$ cross from $s$.  The total offered work is six.}
\label{fig:flow-threshold}
\end{figure}
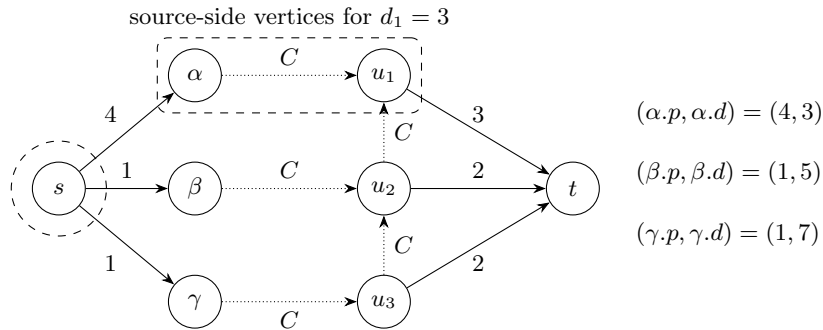

\begin{corollary}[Feasibility and full marginal]
\label{cor:full-marginal}
A set $X$ is feasible if and only if $r(X)=\Proc{X}$.  If $X$ is feasible and
$\alpha\notin X$, then
\begin{equation}
  X+\alpha\text{ is feasible}
  \quad\Longleftrightarrow\quad
  r(X+\alpha)-r(X)=\alpha.p.
  \label{eq:full-marginal}
\end{equation}
Thus the insertion query asks whether the new job contributes its full rank
marginal.
\end{corollary}
\begin{proof}
The $\ell$th term of~\eqref{eq:flow-rank} is at least $\Proc{X}$ if and only
if $\ProcTwo{X}{d_\ell}\leq d_\ell$; the term for $\ell=0$ equals
$\Proc{X}$.  This proves the first assertion.  When $X$ is feasible,
\eqref{eq:full-marginal} is equivalent to
$r(X+\alpha)=\Proc{X+\alpha}$, so the first assertion applied to
$X+\alpha$ proves the second.
\end{proof}

\begin{theorem}[Scheduling polymatroid]
\label{thm:polymatroid}
The function $r:2^S\to\mathbb R_{\geq0}$ is normalized, monotone, and
submodular.  With integral input data it is integer-valued.
\end{theorem}
\begin{proof}
Put $B_\ell=\{\alpha\in S:\alpha.d>d_\ell\}$.  These sets form a decreasing
chain, and
\[
  r(X)=\min_\ell\bigl(d_\ell+\Proc{X\cap B_\ell}\bigr).
\]
Normalization and monotonicity follow from this formula.  Choose minimizing
indices $i$ for $X$ and $j$ for $Y$; by symmetry assume $i\leq j$, so
$B_i\supseteq B_j$.  Job by job,
\[
 \Proc{X\cap B_i}+\Proc{Y\cap B_j}
 \geq\Proc{(X\cap Y)\cap B_i}+\Proc{(X\cup Y)\cap B_j}.
\]
After adding $d_i+d_j$, the right side is at least
$r(X\cap Y)+r(X\cup Y)$.  This proves submodularity.  Integral data make every
term in~\eqref{eq:flow-rank} integral.
\end{proof}

\begin{corollary}[Diminishing marginal capacity]
\label{cor:diminishing}
If $X\subseteq Y\subseteq S$ and $\alpha\notin Y$, then
\[
  r(X+\alpha)-r(X)
  \geq r(Y+\alpha)-r(Y).
\]
In particular, if $X$ and $Y$ are feasible and $\alpha$ cannot be added to
$X$, then it cannot be added to $Y$.
\end{corollary}
\begin{proof}
The displayed inequality is the diminishing-returns form of submodularity.
If $\alpha$ cannot be added to feasible $X$, its marginal is smaller than
$\alpha.p$ by Corollary~\ref{cor:full-marginal}.  Its marginal at $Y$ is no
larger, so it cannot be added to $Y$.
\end{proof}

A minimizing threshold in Theorem~\ref{thm:threshold-cut} is a rejection
certificate: it identifies a deadline prefix whose offered work exceeds its
available machine time.  If every processing time is $1$ and deadlines are
integral, then $r(X)$ is the maximum cardinality of a feasible subset of
$X$; the polymatroid rank specializes to the classical deadline-scheduling
matroid rank.

\subsection{Contraction to an Equal-Processing-Time Tier}

Fix a positive processing time $p$, let
$A\in\Greedy{\LessThan{S}{p}}$, and put $E=\EqualTo{S}{p}$.  For
$Y\subseteq E$, define the contracted rank
\[
  \rho_{A,p}(Y)=r(A\cup Y)-r(A).
\]
Let
\[
  0=\tau_0<\tau_1<\cdots<\tau_q
\]
be the distinct deadlines in $A\cup E$, and define
\begin{equation}
  D_k=\{\alpha\in E:\alpha.d\leq\tau_k\},\qquad
  c_k=\left\lfloor
      \frac{\tau_k-\ProcTwo{A}{\tau_k}}p
      \right\rfloor,\qquad
  b_k=\min_{j\geq k}c_j.
  \label{eq:flow-tier-data}
\end{equation}
Thus $D_0\subseteq\cdots\subseteq D_q=E$ and
$0=b_0\leq\cdots\leq b_q$.  

Since $A$ is feasible, $c_k\geq0$ for every $k$; since $c_0=0$, this
also gives $b_0=0$. Moreover, if $\alpha\in E$ has deadline $\tau_k$,
then
\[
\Level{A}{\alpha}
=
\left\lfloor
\frac{\SlackTwo{A}{\tau_k}}p
\right\rfloor
=
\min_{j\geq k}
\left\lfloor
\frac{\tau_j-\ProcTwo{A}{\tau_j}}p
\right\rfloor
=
b_k.
\]
Thus the level-set constraints of Theorem~\ref{thm:level} are
equivalently
\[
|I\cap D_k|\leq b_k
\qquad(0\leq k\leq q).
\]

\begin{lemma}[Contracted threshold formula]
\label{lem:contracted-threshold}
For every $Y\subseteq E$,
\begin{equation}
  \left\lfloor\frac{\rho_{A,p}(Y)}p\right\rfloor
  =\min_{0\leq k\leq q}
       \bigl(b_k+|Y\setminus D_k|\bigr).
  \label{eq:contracted-threshold}
\end{equation}
\end{lemma}
\begin{proof}
Since $A$ is feasible, $r(A)=\Proc{A}$.  Write
\[
  g_A(d)=d-\ProcTwo{A}{d}.
\]
Subtracting $r(A)$ from~\eqref{eq:flow-rank} gives
\begin{equation}
  \rho_{A,p}(Y)
  =\min_\ell
     \bigl(g_A(d_\ell)
           +p|\{\alpha\in Y:\alpha.d>d_\ell\}|\bigr).
  \label{eq:contracted-cut}
\end{equation}
Between consecutive values of $\tau_0,\ldots,\tau_q$, the set term is
constant and $g_A(d)$ increases with $d$.  Hence the minimum in
\eqref{eq:contracted-cut} is attained at some $\tau_k$.  Taking floors gives
\begin{equation}
  \left\lfloor\frac{\rho_{A,p}(Y)}p\right\rfloor
  =\min_{0\leq k\leq q}
       \bigl(c_k+|Y\setminus D_k|\bigr).
  \label{eq:raw-contracted-rank}
\end{equation}

It remains to replace the raw capacities by their suffix minima.  Since
$b_k\leq c_k$,
\[
  \min_k\bigl(b_k+|Y\setminus D_k|\bigr)
  \leq
  \min_k\bigl(c_k+|Y\setminus D_k|\bigr).
\]
Conversely, for each $k$, choose $j\geq k$ with $b_k=c_j$.  Since
$D_k\subseteq D_j$,
\[
  c_j+|Y\setminus D_j|
  \leq b_k+|Y\setminus D_k|.
\]
The minimum of the raw terms is therefore no larger than every tightened
term.  Taking the minimum over $k$ proves the reverse inequality and
\eqref{eq:contracted-threshold}.
\end{proof}

\begin{theorem}[Tier rank from contracted flow]
\label{thm:contracted-rank}
For every $Y\subseteq E$,
\begin{equation}
  \operatorname{rank}_{\Matroid_{S,p}}(Y)
  =\left\lfloor\frac{\rho_{A,p}(Y)}p\right\rfloor.
  \label{eq:contracted-rank}
\end{equation}
Consequently, the right side is independent of the lower-tier tie order, and
$Y$ is independent in $\Matroid_{S,p}$ if and only if
$\rho_{A,p}(Y)=p|Y|$.
\end{theorem}
\begin{proof}
By Theorem~\ref{thm:level}, a set $I\subseteq E$ is independent in
$\Matroid_{S,p}$ if and only if
\[
  |I\cap D_k|\leq b_k
  \qquad(0\leq k\leq q).
\]
Every independent $I\subseteq Y$ therefore satisfies
\[
  |I|\leq b_k+|Y\setminus D_k|
  \qquad(0\leq k\leq q).
\]
Hence the rank of $Y$ is at most the right side of
\eqref{eq:contracted-threshold}.

Let
\[
  h=\min_{0\leq k\leq q}
       \bigl(b_k+|Y\setminus D_k|\bigr).
\]
The term for $k=0$ shows that $h\leq|Y|$.  Choose $h$ jobs of $Y$ with
latest deadlines, breaking ties arbitrarily.  For each $k$, put
$t=|Y\setminus D_k|$.  If $t\geq h$, all chosen jobs lie outside $D_k$; if
$t<h$, exactly $h-t\leq b_k$ chosen jobs lie in $D_k$.  The chosen set is
independent and has size $h$.  Lemma~\ref{lem:contracted-threshold} now
proves~\eqref{eq:contracted-rank}.

Finally, the $d_0=0$ term in~\eqref{eq:contracted-cut} gives
$\rho_{A,p}(Y)\leq p|Y|$.  Thus the rank equals $|Y|$ if and only if
$\rho_{A,p}(Y)=p|Y|$.
\end{proof}

\begin{corollary}[SJF is tierwise Edmonds greedy]
\label{cor:tierwise-edmonds}
During tier $p$, a job is accepted if and only if it raises the rank of
$\Matroid_{S,p}$ by one.  Thus an SJF scan is an ordered sequence of ordinary
matroid-greedy phases, one for each processing-time tier.
\end{corollary}
\begin{proof}
Let $Y$ be the already accepted slice of tier $p$.  Since $Y$ is independent,
Theorem~\ref{thm:contracted-rank} gives
$\rho_{A,p}(Y)=p|Y|$.  By Corollary~\ref{cor:full-marginal}, the next job
$\alpha$ is accepted if and only if
\[
  \rho_{A,p}(Y+\alpha)=\rho_{A,p}(Y)+p.
\]
Theorem~\ref{thm:contracted-rank} identifies this condition with a unit
increase in the rank of $\Matroid_{S,p}$.
\end{proof}

\subsection{Fractional Optimization}

Let
\[
  P(r)=\{x\in\mathbb R_{\geq0}^{S}:x(Y)\leq r(Y)
         \text{ for all }Y\subseteq S\},
  \qquad
  x(Y)=\sum_{\alpha\in Y}x_\alpha.
\]
Here $x_\alpha$ is the amount of processing admitted from job $\alpha$; a
point of $P(r)$ may admit only part of a job.

\begin{corollary}[Linear optimization over the scheduling polymatroid]
\label{cor:poly-greedy}
Given nonnegative weights $w_\alpha$, a maximizer of $w^{\mathsf T}x$ over
$P(r)$ can be computed in $O(n\log n)$ time. For integral input, the returned
vector is integral in units of processing time, although it may allocate only
part of a job.
\end{corollary}
\begin{proof}
Order the jobs so that
$w_{\pi(1)}\geq\cdots\geq w_{\pi(n)}$, put
$U_i=\{\pi(1),\ldots,\pi(i)\}$, and set
\[
  x_{\pi(i)}=r(U_i)-r(U_{i-1}).
\]
Edmonds' polymatroid greedy theorem proves optimality
\cite{Edmonds70,Fujishige05}.  The marginals can be computed by maintaining
\[
  z_\ell=d_\ell+
  \sum_{\substack{\alpha\in U_i\\\alpha.d>d_\ell}}\alpha.p.
\]
Inserting a job of deadline $d_j$ adds its processing time to exactly
$z_0,\ldots,z_{j-1}$, and $r(U_i)=\min_\ell z_\ell$.  Coördinate
compression and a lazy segment tree support each prefix addition and global
minimum in $O(\log n)$ time.  Integrality follows from
Theorem~\ref{thm:polymatroid}.
\end{proof}

The tree and the flow rank are two resolutions of the same deadline-prefix
capacity.  The tree implements the full-marginal decision dynamically; the
flow rank packages all thresholds into one submodular function.  The
polymatroid does not replace the exchange argument: its greedy algorithm may
allocate only part of a job, and the family of feasible job sets need not be
a matroid.  Its discrete role is exact but tierwise.  After contraction by
the shorter jobs and division by $p$, its rank becomes the nested-matroid
rank tested during the processing-time-$p$ phase of SJF.


\section{Concluding Remarks}
\label{sec:concs}
We presented both an efficient implementation and a structural explanation of Moore's greedy algorithm.  By equipping a deadline-mass map with a two-number subtree summary, our augmented BST evaluates feasibility in worst-case $O(\log n)$ time, yielding a simple $O(n\log n)$-time implementation. We also proved that, over all SJF tie orders, the first $i$ accepted jobs are exactly the feasible $i$-sets of minimum total processing time.

At a deeper level, this behavior is governed by a nested matroidal structure. Although the overall family of feasible job-sets is not a matroid, the jobs chosen within each processing-time tier form a base of a nested matroid, independent of earlier tiers. An SJF scan decomposes into a sequence of ordinary matroid-greedy phases. The network flow formulation unifies this view with the implementation: its polymatroid rank measures the processing that can be accommodated by the deadline-prefix capacities, and the corresponding tier matroids arise from residual versions of this rank. The augmented BST, the tier matroids, and the flow polymatroid therefore describe the same underlying deadline-prefix structure.

\label{lastmain}
\bibliographystyle{plain}
\bibliography{mybib}

\appendix

\section{A Variant of the Augmented BST Data Structure}
\label{sec:variant}

In this section, we describe a variant of the augmented BST data
structure $T$ presented in Section~\ref{sec:augment}. This variant is
slightly more complicated to describe but has the advantage that the
tree is only modified when a job $\alpha$ is added to the current
feasible set $X$.  The fields maintained at each node and the
invariants relating $X$ and $T$ are the same as in
Section~\ref{sec:augment}.  This variant is based on two operations: a
read-only logarithmic-time $\TreeAdmits{T}{d}{p}$ operation and the
logarithmic-time $\TreeAdd{T}{d}{p}$ operation of
Section~\ref{sec:augment}. (The constant-time $\TreeSlack{T}$
operation is no longer needed.)  To determine whether a given job
$\alpha$ can be added to the current feasible set $X$ to obtain a
larger feasible set, we call $\TreeAdmits{T}{\alpha.d}{\alpha.p}$,
which returns a Boolean value indicating whether
$\SlackTwo{X}{\alpha.d}\geq\alpha.p$. Thus, by
Lemma~\ref{lem:feasibleToAdd}, $\TreeAdmits{T}{\alpha.d}{\alpha.p}$
returns $\True$ if and only if $X+\alpha$ is feasible.  If
$\TreeAdmits{T}{\alpha.d}{\alpha.p}$ returns $\True$, we add $\alpha$
to $X$ and call $\TreeAdd{T}{\alpha.d}{\alpha.p}$ to re-establish the
invariants with respect to the new feasible set $X$.

Note that in this variant, we never need to ``undo'' the effect of a
call $\TreeAdmits{T}{\alpha.d}{\alpha.p}$ as we did in
Section~\ref{sec:augment}; thus the $p$-value passed to any call to
$\TreeAdmits{T}{d}{p}$ is positive.  Below we complete our discussion
of the variant by presenting a read-only logarithmic-time
implementation of the $\TreeAdmits{T}{d}{p}$ operation.  Our
pseudocode is given in Algorithm~\ref{alg:admits}. 

The operation runs in $O(\log n)$ time: each iteration takes
constant time and either terminates or moves from a node to one of
its children.

\begin{algorithm}[H]
\caption{The read-only admits operation}
\label{alg:admits}
\begin{algorithmic}[1]
\Function{$T.\AdmitsOp$}{$d,p$}
  \State $u\gets\mbox{root of }T$
  \While{$u\not=\Nil$}
    \If{$u.d<d$}
       \State $p\gets p+u.\ell.x+u.p$
       \label{line:accumulate}
       \State $u\gets u.r$
       \label{line:right}
    \ElsIf{$\min(u.d,u.r.y) \geq p+u.\ell.x+u.p$}
       \State $u\gets u.\ell$
       \label{line:left}
    \Else
       \State \Return $\False$
       \label{line:false}
    \EndIf
  \EndWhile
  \State \Return $p\le d$
\EndFunction
\end{algorithmic}
\end{algorithm}

The key invariant associated with the while loop of
Algorithm~\ref{alg:admits} is ``the correct return value is
$\SlackTwo{X_u}{d}\geq p$''. It is easy to see that this invariant
holds the first time the while loop is reached. Assume that the
invariant holds at the beginning of a general iteration of the while
loop. There are two cases to consider.

Case~1: $u=\Nil$. In this case, the loop terminates and we need to
prove that the return value is equivalent to $\SlackTwo{X_u}{d}\geq
p$. This holds since $\SlackTwo{X_u}{d}=d$ when $u$ is $\Nil$, and the
algorithm returns $p\leq d$.

Case~2: $u\not=\Nil$. In this case, the body of the loop is executed
and we need to prove that the invariant holds before the next
iteration. We consider two subcases.

Case~2.1: $u.d<d$. In this case,
$\SlackTwo{X_u}{d}=\SlackTwo{X_{u.r}}{d}-u.\ell.x-u.p$ and hence
$\SlackTwo{X_u}{d}\geq p$ if and only if $\SlackTwo{X_{u.r}}{d}\geq
p+u.\ell.x+u.p$. Thus the assignment statements of
Lines~\ref{line:accumulate} and~\ref{line:right} re-establish the
invariant before the next iteration of the while loop.

Case~2.2: $u.d\geq d$.  In this case,
\[
\SlackTwo{X_u}{d}
=\min(\SlackTwo{X_{u.\ell}}{d},\min(u.d,u.r.y)-u.\ell.x-u.p)
\]
and we consider the following two subcases.

Case~2.2.1: $\min(u.d,u.r.y)\geq p+u.\ell.x+u.p$. In this case,
$\min(u.d,u.r.y)-u.\ell.x-u.p \geq p$ and hence $\SlackTwo{X_u}{d}\geq
p$ if and only if $\SlackTwo{X_{u.\ell}}{d}\geq p$. Thus the
assignment statement of Line~\ref{line:left} re-establishes the
invariant before the next iteration of the while loop.

Case~2.2.2: $\min(u.d,u.r.y)<p+u.\ell.x+u.p$. In this case,
$\SlackTwo{X_u}{d}<p$, and hence it is correct to return $\False$ in
Line~\ref{line:false}.

\end{document}